\documentclass{article}

\usepackage{titling}

\usepackage{times}
\usepackage{graphicx} 
\usepackage[lofdepth,lotdepth]{subfig}

\usepackage{natbib}

\usepackage{algpseudocode}
\usepackage{algorithm}

\usepackage{amssymb}
\usepackage{amsthm}
\usepackage{amsmath}

\newtheorem{thm}{Theorem}
\newtheorem*{thm*}{Theorem}
\newtheorem{defn}[thm]{Definition}
\newtheorem{lem}[thm]{Lemma}
\newtheorem*{lem*}{Lemma}

\newtheorem{innercustomthm}{Theorem}
\newenvironment{customthm}[1]
  {\renewcommand\theinnercustomthm{#1}\innercustomthm}
  {\endinnercustomthm}
\newtheorem{innercustomlem}{Lemma}
\newenvironment{customlem}[1]
  {\renewcommand\theinnercustomlem{#1}\innercustomlem}
  {\endinnercustomlem}

\usepackage{hyperref}

\usepackage[accepted]{icml2016}

\icmltitlerunning{Fast $k$-Nearest Neighbour Search via Dynamic Continuous Indexing}

\begin{document} 

\twocolumn[
\icmltitle{Fast $k$-Nearest Neighbour Search via Dynamic Continuous Indexing}

\icmlauthor{Ke Li}{ke.li@eecs.berkeley.edu}
\icmlauthor{Jitendra Malik}{malik@eecs.berkeley.edu}
\icmladdress{University of California, Berkeley, CA 94720, United States}

\icmlkeywords{nearest neighbour search, randomized algorithm, one-dimensional random projection, space partitioning, curse of dimensionality}

\vskip 0.3in
]

\begin{abstract}
Existing methods for retrieving $k$-nearest neighbours suffer from the curse of dimensionality. We argue this is caused in part by inherent deficiencies of space partitioning, which is the underlying strategy used by most existing methods. We devise a new strategy that avoids partitioning the vector space and present a novel randomized algorithm that runs in time linear in dimensionality of the space and sub-linear in the intrinsic dimensionality and the size of the dataset and takes space constant in dimensionality of the space and linear in the size of the dataset. The proposed algorithm allows fine-grained control over accuracy and speed on a per-query basis, automatically adapts to variations in data density, supports dynamic updates to the dataset and is easy-to-implement. We show appealing theoretical properties and demonstrate empirically that the proposed algorithm outperforms locality-sensitivity hashing (LSH) in terms of approximation quality, speed and space efficiency. 
\end{abstract} 

\section{Introduction}

The $k$-nearest neighbour method is commonly used both as a classifier and as subroutines in more complex algorithms in a wide range domains, including machine learning, computer vision, graphics and robotics. Consequently, finding a fast algorithm for retrieving nearest neighbours has been a subject of sustained interest among the artificial intelligence and the theoretical computer science communities alike. Work over the past several decades has produced a rich collection of algorithms; however, they suffer from one key shortcoming: as the ambient or intrinsic dimensionality\footnote{We use \emph{ambient dimensionality} to refer to the original dimensionality of the space containing the data in order to differentiate it from \emph{intrinsic dimensionality}, which measures density of the dataset and will be defined precisely later.} increases, the running time and/or space complexity grows rapidly; this phenomenon is often known as the curse of dimensionality. Finding a solution to this problem has proven to be elusive and has been conjectured to be fundamentally impossible \cite{minsky1969perceptrons}. In this era of rapid growth in both the volume and dimensionality of data, it has become increasingly important to devise a fast algorithm that scales to high dimensional space. 

We argue that the difficulty in overcoming the curse of dimensionality stems in part from inherent deficiencies in space partitioning, the strategy that underlies most existing algorithms. Space partitioning is a divide-and-conquer strategy that partitions the vector space into a finite number of cells and keeps track of the data points that each cell contains. At query time, exhaustive search is performed over all data points that lie in cells containing the query point. This strategy forms the basis of most existing algorithms, including $k$-d trees \cite{bentley1975multidimensional} and locality-sensitive hashing (LSH) \cite{indyk1998approximate}. 

While this strategy seems natural and sensible, it suffers from critical deficiencies as the dimensionality increases. Because the volume of any region in the vector space grows exponentially in the dimensionality, either the number or the size of cells must increase exponentially in the number of dimensions, which tends to lead to exponential time or space complexity in the dimensionality. In addition, space partitioning limits the algorithm's ``field of view" to the cell containing the query point; points that lie in adjacent cells have no hope of being retrieved. Consequently, if a query point falls near cell boundaries, the algorithm will fail to retrieve nearby points that lie in adjacent cells. Since the number of such cells is exponential in the dimensionality, it is intractable to search these cells when dimensionality is high. One popular approach used by LSH and spill trees \cite{liu2004investigation} to mitigate this effect is to partition the space using overlapping cells and search over all points that lie in any of the cells that contain the query point. Because the ratio of surface area to volume grows in dimensionality, the number of overlapping cells that must be used increases in dimensionality; as a result, the running time or space usage becomes prohibitively expensive as dimensionality increases. Further complications arise from variations in data density across different regions of the space. If the partitioning is too fine, most cells in sparse regions of the space will be empty and so for a query point that lies in a sparse region, no points will be retrieved. If the partitioning is too coarse, each cell in dense regions of the space will contain many points and so for a query point that lies in a dense region, many points that are not the true nearest neighbours must be searched. This phenomenon is notably exhibited by LSH, whose performance is highly sensitive to the choice of the hash function, which essentially defines an implicit partitioning. A good partitioning scheme must therefore depend on the data; however, such data-dependent partitioning schemes would require possibly expensive preprocessing and prohibit online updates to the dataset. These fundamental limitations of space partitioning beg an interesting question: is it possible to devise a strategy that does not partition the space, yet still enables fast retrieval of nearest neighbours?

In this paper, we present a new strategy for retrieving $k$-nearest neighbours that avoids discretizing the vector space, which we call dynamic continuous indexing (DCI). Instead of partitioning the space into discrete cells, we construct continuous indices, each of which imposes an ordering on data points such that closeness in position serves as an approximate indicator of proximity in the vector space. The resulting algorithm runs in time linear in ambient dimensionality and sub-linear in intrinsic dimensionality and the size of the dataset, while only requiring space constant in ambient dimensionality and linear in the size of the dataset. Unlike existing methods, the algorithm allows fine-grained control over accuracy and speed at query time and adapts to varying data density on-the-fly while permitting dynamic updates to the dataset. Furthermore, the algorithm is easy-to-implement and does not rely on any complex or specialized data structure. 

\section{Related Work}

Extensive work over the past several decades has produced a rich collection of algorithms for fast retrieval of $k$-nearest neighbours. Space partitioning forms the basis of the majority of these algorithms. Early approaches store points in deterministic tree-based data structures, such as $k$-d trees \cite{bentley1975multidimensional}, R-trees \cite{guttman1984r} and X-trees \cite{berchtold1996x,berchtold1998fast}, which effectively partition the vector space into a hierarchy of half-spaces, hyper-rectangles or Voronoi polygons. These methods achieve query times that are logarithmic in the number of data points and work very well for low-dimensional data. Unfortunately, their query times grow exponentially in ambient dimensionality because the number of leaves in the tree that need to be searched increases exponentially in ambient dimensionality; as a result, for high-dimensional data, these algorithms become slower than exhaustive search. More recent methods like spill trees~\cite{liu2004investigation}, RP trees~\cite{dasgupta2008random} and virtual spill trees~\cite{dasgupta2015randomized} extend these approaches by randomizing the dividing hyperplane at each node. Unfortunately, the number of points in the leaves increases exponentially in the intrinsic dimensionality of the dataset. 

In an effort to tame the curse of dimensionality, researchers have considered relaxing the problem to allow $\epsilon$-approximate solutions, which can contain any point whose distance to the query point differs from that of the true nearest neighbours by at most a factor of $(1+\epsilon)$. Tree-based methods \cite{arya1998optimal} have been proposed for this setting; unfortunately, the running time still exhibits exponential dependence on dimensionality. Another popular method is locality-sensitive hashing (LSH) \cite{indyk1998approximate,datar2004locality}, which relies on a hash function that implicitly defines a partitioning of the space. Unfortunately, LSH struggles on datasets with varying density, as cells in sparse regions of the space tend to be empty and cells in dense regions tend to contain a large number of points. As a result, it fails to return any point on some queries and requires a long time on some others. This motivated the development of data-dependent hashing schemes based on $k$-means \cite{pauleve2010locality} and spectral partitioning \cite{weiss2009spectral}. Unfortunately, these methods do not support dynamic updates to the dataset or provide correctness guarantees. Furthermore, they incur a significant pre-processing cost, which can be expensive on large datasets. Other data-dependent algorithms outside the LSH framework have also been proposed, such as \cite{fukunaga1975branch,brin1995near,nister2006scalable,wang2011fast}, which work by constructing a hierarchy of clusters using $k$-means and can be viewed as performing a highly data-dependent form of space partitioning. For these algorithms, no guarantees on approximation quality or running time are known. 

One class of methods~\cite{orchard1991fast,arya1993approximate,clarkson1999nearest,karger2002finding} that does not rely on space partitioning uses a local search strategy. Starting from a random data point, these methods iteratively find a new data point that is closer to the query than the previous data point. Unfortunately, the performance of these methods deteriorates in the presence of significant variations in data density, since it may take very long to navigate a dense region of the space, even if it is very far from the query. Other methods like navigating nets~\cite{krauthgamer2004navigating}, cover trees~\cite{beygelzimer2006cover} and rank cover trees~\cite{houle2015rank} adopt a coarse-to-fine strategy. These methods work by maintaining coarse subsets of points at varying scales and progressively searching the neighbourhood of the query with decreasing radii at increasingly finer scales. Sadly, the running times of these methods again exhibit exponential dependence on the intrinsic dimensionality. 

We direct interested readers to~\cite{clarkson2006nearest} for a comprehensive survey of existing methods. 

\section{Algorithm}

\begin{algorithm}
\footnotesize
\caption{Algorithm for data structure construction}
\label{alg_construct}
\begin{algorithmic}
\Require A dataset $D$ of $n$ points $p^{1},\ldots,p^{n}$, the number of simple indices $m$ that constitute a composite index and the number of composite indices $L$
\Function{Construct}{$D,m,L$}
    \State $\{u_{jl}\}_{j \in [m], l \in [L]} \gets mL$ random unit vectors in $\mathbb{R}^{d}$
    \State $\{T_{jl}\}_{j \in [m], l \in [L]} \gets mL$ empty binary search trees or skip 
    \State \qquad\qquad\qquad\qquad\, lists
    \For{$j = 1$ \textbf{to} $m$}
        \For{$l = 1$ \textbf{to} $L$}
            \For{$i = 1$ \textbf{to} $n$}
                \State $\overline{p}^{i}_{jl} \gets \langle p^{i},u_{jl}\rangle$
                \State Insert $(\overline{p}^{i}_{jl},i)$ into $T_{jl}$ with $\overline{p}^{i}_{jl}$ being the key and 
                \State \; $i$ being the value
            \EndFor
        \EndFor
    \EndFor
    \State \Return $\{(T_{jl},u_{jl})\}_{j \in [m], l \in [L]}$
\EndFunction
\end{algorithmic}
\normalsize
\end{algorithm}

The proposed algorithm relies on the construction of continuous indices of data points that support both fast searching and online updates. To this end, we use one-dimensional random projections as basic building blocks and construct $mL$ simple indices, each of which orders data points by their projections along a random direction. Such an index has both desired properties: data points can be efficiently retrieved and updated when the index is implemented using standard data structures for storing ordered sequences of scalars, like self-balancing binary search trees or skip lists. This ordered arrangement of data points exploits an important property of the $k$-nearest neighbour search problem that has often been overlooked: it suffices to construct an index that approximately preserves the relative \emph{order} between the true $k$-nearest neighbours and the other data points in terms of their distances to the query point without necessarily preserving all pairwise distances. This observation enables projection to a much lower dimensional space than the Johnson-Lindenstrauss transform \cite{johnson1984extensions}. We show in the following section that with high probability, one-dimensional random projection preserves the relative order of two points whose distances to the query point differ significantly regardless of the ambient dimensionality of the points. 

\begin{algorithm}
\footnotesize
\caption{Algorithm for $k$-nearest neighbour retrieval}
\label{alg_query}
\begin{algorithmic}
\Require Query point $q$ in $\mathbb{R}^{d}$, binary search trees/skip lists and their associated projection vectors $\{(T_{jl},u_{jl})\}_{j \in [m], l \in [L]}$, and maximum tolerable failure probability $\epsilon$
\Function{Query}{$q,\{(T_{jl},u_{jl})\}_{j,l},\epsilon$}
    \State $C_{l} \gets$ array of size $n$ with entries initialized to 0\; $\forall l \in [L]$
    \State $\overline{q}_{jl} \gets \langle q,u_{jl}\rangle\; \forall j \in [m], l \in [L]$
    \State $S_{l} \gets \emptyset \; \forall l \in [L]$
    \For{$i = 1$ \textbf{to} $n$}
        \For{$l = 1$ \textbf{to} $L$}
            \For{$j = 1$ \textbf{to} $m$}
                \State $(\overline{p}_{jl}^{(i)},h_{jl}^{(i)}) \gets $ the node in $T_{jl}$ whose key is the 
                \State \qquad\qquad\qquad $i^{\mathrm{th}}$ closest to $\overline{q}_{jl}$
                \State $C_{l}[h_{jl}^{(i)}] \gets C_{l}[h_{jl}^{(i)}] + 1$
            \EndFor
            \For{$j = 1$ \textbf{to} $m$}
                \If{$C_{l}[h_{jl}^{(i)}] = m$}
                    \State $S_{l} \gets S_{l} \cup \{h_{jl}^{(i)}\}$
                \EndIf
            \EndFor
        \EndFor
        \If{\textit{IsStoppingConditionSatisfied}($i,S_{l},\epsilon$)}
            \State \textbf{break}
        \EndIf
    \EndFor
    \State \Return $k$ points in $\bigcup_{l\in[L]}S_{l}$ that are the closest in 
    \State \qquad\;\;\; Euclidean distance in $\mathbb{R}^{d}$ to $q$
\EndFunction
\end{algorithmic}
\normalsize
\end{algorithm}

We combine each set of $m$ simple indices to form a composite index in which points are ordered by the maximum difference over all simple indices between the positions of the point and the query in the simple index. The composite index enables fast retrieval of a small number of data points, which will be referred to as candidate points, that are close to the query point along several random directions and therefore are likely to be truly close to the query. The composite indices are not explicitly constructed; instead, each of them simply keeps track of the number of its constituent simple indices that have encountered any particular point and returns a point as soon as all its constituent simple indices have encountered that point. 

At query time, we retrieve candidate points from each composite index one by one in the order they are returned until some stopping condition is satisfied, while omitting points that have been previously retrieved from other composite indices. Exhaustive search is then performed over candidate points retrieved from all $L$ composite indices to identify the subset of $k$ points closest to the query. Please refer to Algorithms \ref{alg_construct} and \ref{alg_query} for a precise statement of the construction and querying procedures. 

\begin{figure*}[t]
    \centering
    \subfloat[]{
        \includegraphics[width=0.33\textwidth]{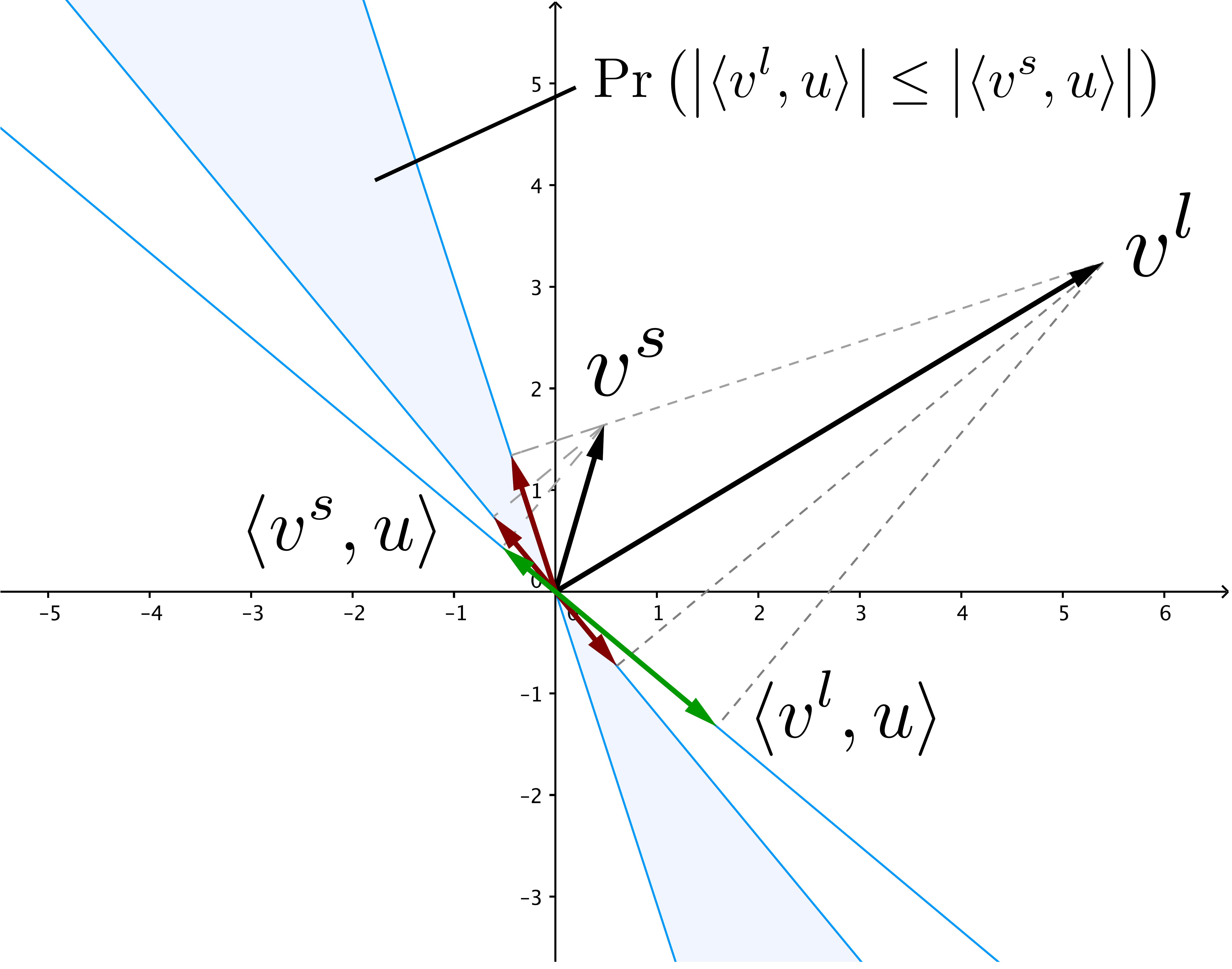}
        \label{fig:2d_proj}
    }
    \subfloat[]{
        \includegraphics[width=0.33\textwidth]{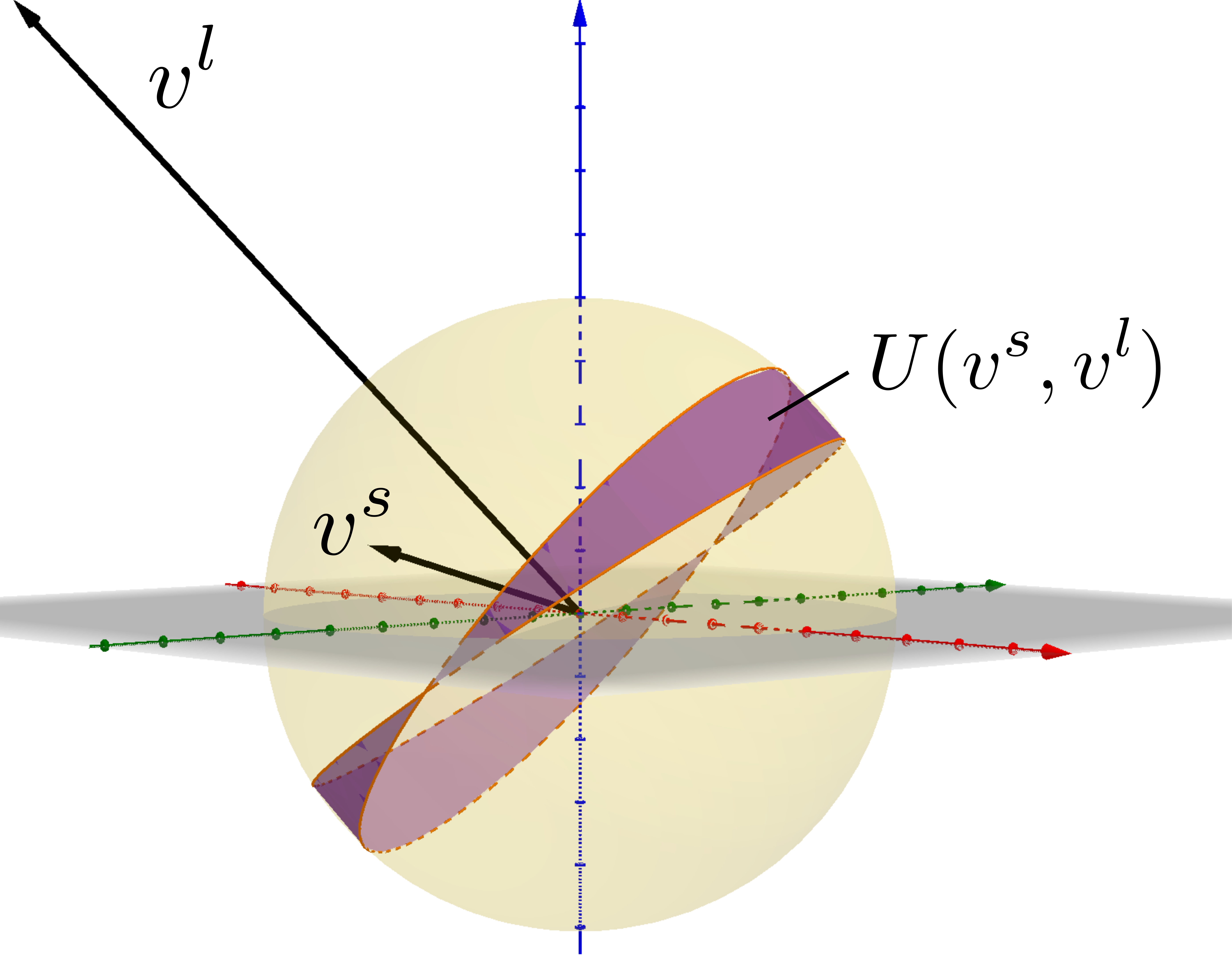}
        \label{fig:3d_proj_one_vec}
    }
    \subfloat[]{
        \includegraphics[width=0.33\textwidth]{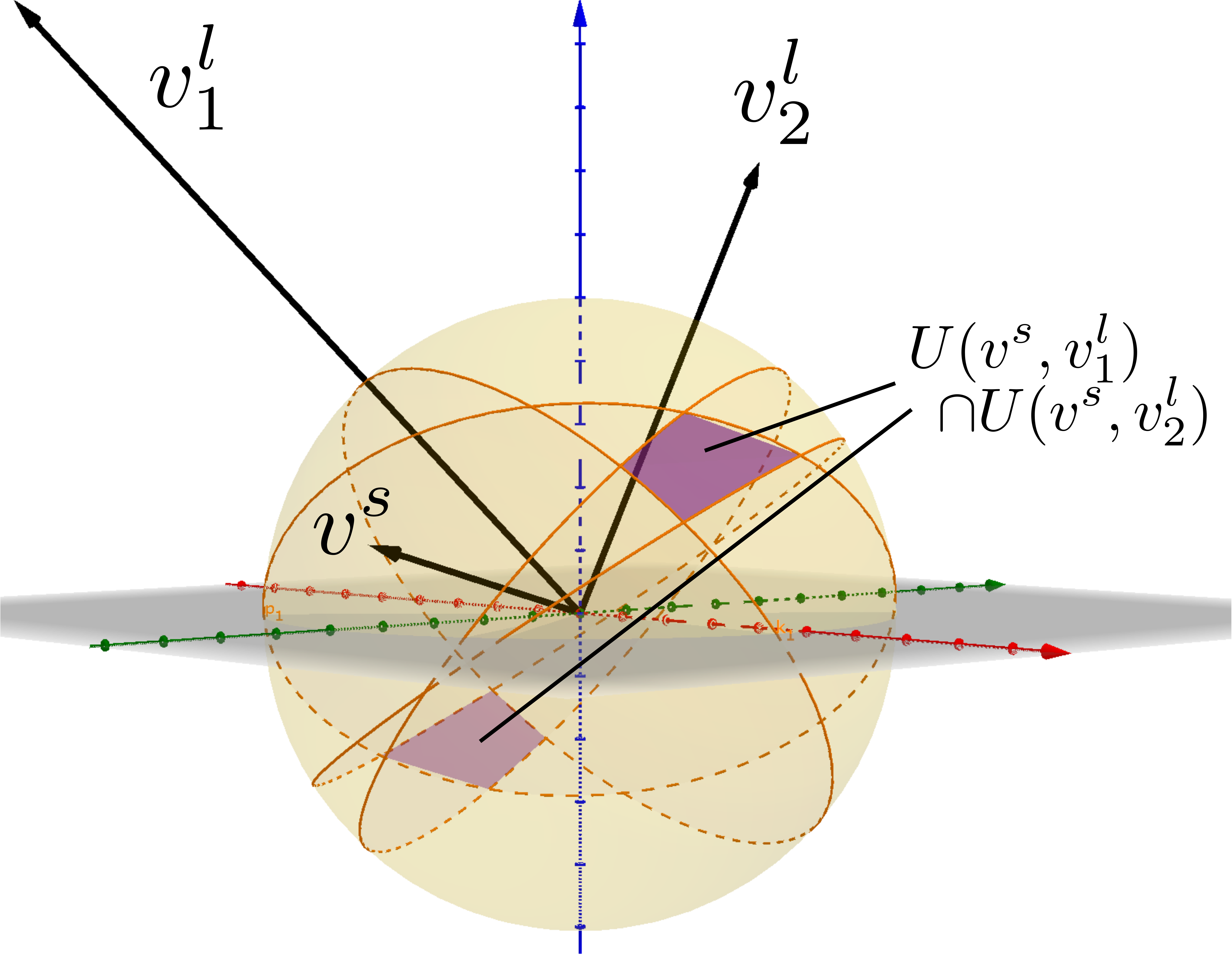}
        \label{fig:3d_proj_two_vec}
    }
    \caption{(a) Examples of order-preserving (shown in green) and order-inverting (shown in red) projection directions. Any projection direction within the shaded region inverts the relative order of the vectors by length under projection, while any projection directions outside the region preserves it. The size of the shaded region depends on the ratio of the lengths of the vectors. (b) Projection vectors whose endpoints lie in the shaded region would be order-inverting. (c) Projection vectors whose endpoints lie in the shaded region would invert the order of both long vectors relative to the short vector. Best viewed in colour.}
\end{figure*}

Because data points are retrieved according to their positions in the composite index rather than the regions of space they lie in, the algorithm is able to automatically adapt to changes in data density as dynamic updates are made to the dataset without requiring any pre-processing to estimate data density at construction time. Also, unlike existing methods, the number of retrieved candidate points can be controlled on a per-query basis, enabling the user to easily trade off accuracy against speed. We develop two versions of the algorithm, a data-independent and a data-dependent version, which differ in the stopping condition that is used. In the former, the number of candidate points is indirectly preset according to the global data density and the maximum tolerable failure probability; in the latter, the number of candidate points is chosen adaptively at query time based on the local data density in the neighbourhood of the query. We analyze the algorithm below and show that its failure probability is independent of ambient dimensionality, its running time is linear in ambient dimensionality and sub-linear in intrinsic dimensionality and the size of the dataset and its space complexity is independent of ambient dimensionality and linear in the size of the dataset. 

\subsection{Properties of Random 1D Projection}

First, we examine the effect of projecting $d$-dimensional vectors to one dimension, which motivates its use in the proposed algorithm. We are interested in the probability that a distant point appears closer than a nearby point under projection; if this probability is low, then each simple index approximately preserves the order of points by distance to the query point. If we consider displacement vectors between the query point and data points, this probability is then is equivalent to the probability of the lengths of these vectors inverting under projection. 

\begin{lem}
Let $v^{l},v^{s} \in \mathbb{R}^{d}$ such that $\left\Vert v^{l} \right\Vert _{2} > \left\Vert v^{s}\vphantom{v^{l}} \right\Vert _{2}$, and $u \in \mathbb{R}^{d}$ be a unit vector drawn uniformly at random. Then the probability of $v^{s}\vphantom{v^{l}}$ being at least as long as $v^{l}$ under projection $u$ is at most $1-\frac{2}{\pi}\cos^{-1}\left(\left\Vert v^{s}\vphantom{v^{l}}\right\Vert _{2} / \left\Vert v^{l}\right\Vert _{2}\right)$. 
\end{lem}
\begin{proof}
Assuming that $v^{l}$ and $v^{s}$ are not collinear, consider the two-dimensional subspace spanned by $v^{l}$ and $v^{s}$, which we will denote as $P$. (If $v^{l}$ and $v^{s}$ are collinear, we define $P$ to be the subspace spanned by $v^{l}$ and an arbitrary vector that's linearly independent of $v^{l}$.) For any vector $w$, we use $w^{\parallel}$ and $w^{\perp}$ to denote the components of $w$ in $P$ and $P^{\perp}$ such that $w=w^{\parallel}+w^{\perp}$. For $w \in \{v^{s}\vphantom{v^{l}},v^{l}\}$, because $w^{\perp}=0$, $\langle w, u \rangle = \langle w, u^{\parallel} \rangle$. So, we can limit our attention to $P$ for this analysis. We parameterize $u^{\parallel}$ in terms of its angle relative to $v^{l}$, which we denote as $\theta$. Also, we denote the angle of $u^{\parallel}$ relative to $v^{s}$ as $\psi$. Then,

\vspace{-10pt}
\footnotesize
\begin{align*}
 &\mathrm{Pr}\left(\left|\langle v^{l},u\rangle\right| \leq \left|\langle v^{s}\vphantom{v^{l}},u\rangle\right|\right) \\
 = \;& \mathrm{Pr}\left(\left|\langle v^{l},u^{\parallel}\rangle\right| \leq \left|\langle v^{s}\vphantom{v^{l}},u^{\parallel}\rangle\right|\right) \\
 = \;& \mathrm{Pr}\left(\left\Vert v^{l} \right\Vert _{2} \left\Vert u^{\parallel} \right\Vert _{2} \left|\cos\theta\right| \leq \left\Vert v^{s}\vphantom{v^{l}} \right\Vert _{2} \left\Vert u^{\parallel}\right\Vert _{2}\left|\cos\psi\right|\right) \\
 \leq \;& \mathrm{Pr}\left(\left|\cos\theta\right| \leq \frac{\left\Vert v^{s} \right\Vert _{2}}{\left\Vert v^{l} \right\Vert _{2}}\right) \\
 = \;& 2\mathrm{Pr}\left(\theta\in\left[\cos^{-1}\left(\frac{\left\Vert v^{s} \right\Vert _{2}}{\left\Vert v^{l} \right\Vert _{2}}\right),\pi-\cos^{-1}\left(\frac{\left\Vert v^{s} \right\Vert _{2}}{\left\Vert v^{l} \right\Vert _{2}}\right)\right]\right) \\
 = \;& 1- \frac{2}{\pi}\cos^{-1}\left(\frac{\left\Vert v^{s} \right\Vert _{2}}{\left\Vert v^{l} \right\Vert _{2}}\right)
\end{align*}
\normalsize
\end{proof}

Observe that if $\left|\langle v^{l},u\rangle\right|\leq\left|\langle v^{s}\vphantom{v^{l}},u\rangle\right|$, the relative order of $v^{l}$ and $v^{s}$ by their lengths would be inverted when projected along $u$. This occurs when $u^{\parallel}$ is close to orthogonal to $v^{l}$, which is illustrated in Figure \ref{fig:2d_proj}. Also note that the probability of inverting the relative order of $v^{l}$ and $v^{s}$ is small when $v^{l}$ is much longer than $v^{s}$. On the other hand, this probability is high when $v^{l}$ and $v^{s}$ are similar in length, which corresponds to the case when two data points are almost equidistant to the query point. So, if we consider a sequence of vectors ordered by length, applying random one-dimensional projection will likely perturb the ordering locally, but will preserve the ordering globally. 

Next, we build on this result to analyze the order-inversion probability when there are more than two vectors. Consider the sample space $B = \left\{ \left. u \in \mathbb{R}^{d} \right| \left\Vert u \right\Vert _{2} = 1 \right\}$ and the set $U(v^{s},v^{l})=\left\{ u\in B\left| \; \left|\cos\theta\right|\leq\left\Vert v^{s}\vphantom{v^{l}}\right\Vert _{2}/\left\Vert v^{l}\right\Vert _{2}\right.\right\}$, which is illustrated in Figure \ref{fig:3d_proj_one_vec}, where $\theta$ is the angle between $u^{\parallel}$ and $v^{l}$. If we use $\mbox{area}(U)$ to denote the area of the region formed by the endpoints of all vectors in the set $U$, then we can rewrite the above bound on the order-inversion probability as:

\vspace{-10pt}
\footnotesize
\begin{align*}
\mathrm{Pr}\left(\left|\langle v^{l},u\rangle\right|\leq\left|\langle v^{s}\vphantom{v^{l}},u\rangle\right|\right) &\leq \mathrm{Pr}\left(u \in U(v^{s},v^{l}) \right) \\
&= \frac{\mbox{area}(V_{l}^{s})}{\mbox{area}(B)} \\
&= 1-\frac{2}{\pi}\cos^{-1}\left(\frac{\left\Vert v^{s} \right\Vert _{2}}{\left\Vert v^{l} \right\Vert _{2}}\right)
\end{align*}
\normalsize

\begin{lem}
\label{lem:single_order_under_proj}
Let $\left\{ v^{l}_{i}\right\} _{i=1}^{N}$ be a set of vectors such that $\left\Vert v^{l}_{i}\right\Vert _{2}>\left\Vert v^{s}\vphantom{v^{l}} \right\Vert _{2} \; \forall i\in[N]$. Then the probability that there is a subset of $k'$ vectors from $\left\{ v^{l}_{i}\right\} _{i=1}^{N}$ that are all not longer than $v^{s}$ under projection is at most $\frac{1}{k'}\sum_{i=1}^{N}\left(1-\frac{2}{\pi}\cos^{-1}\left(\left\Vert v^{s}\vphantom{v^{l}} \right\Vert _{2} / \left\Vert v^{l}_{i}\right\Vert _{2} \right)\right)$. 
Furthermore, if $k' = N$, this probability is at most $\min_{i\in[N]}\left\{ 1-\frac{2}{\pi}\cos^{-1}\left(\left\Vert v^{s}\vphantom{v^{l}}\right\Vert _{2} / \left\Vert v^{l}_{i}\right\Vert _{2} \right) \right\}$. 
\end{lem}
\begin{proof}
For a given subset $I \subseteq [N]$ of size $k'$, the probability that all vectors indexed by elements in $I$ are not longer than $v^{s}$ under projection $u$ is at most $\mathrm{Pr}\left(u \in \bigcap_{i\in I}U(v^{s},v^{l}_{i})\right)=\mbox{area}\left(\bigcap_{i\in I} U(v^{s},v^{l}_{i})\right) / \mbox{area}\left(B\right)$. So, the probability that this occurs on some subset $I$ is at most $\mathrm{Pr}\left(u \in \bigcup_{I \subseteq [N]:|I|=k'}\bigcap_{i\in I}U(v^{s},v^{l}_{i})\right) = \mbox{area}\left(\bigcup_{I\subseteq[N]:|I|=k'}\bigcap_{i\in I}U(v^{s},v^{l}_{i})\right)/\mbox{area}\left(B\right)$.

Observe that each point in $\bigcup_{I\subseteq[N]:|I|=k'}\bigcap_{i\in I}U(v^{s},v^{l}_{i})$ must be covered by at least $k'$ $U(v^{s},v^{l}_{i})$'s. So,
\footnotesize
\[
k' \cdot \mbox{area}\left(\bigcup_{I\subseteq[N]:|I|=k'}\bigcap_{i\in I} U(v^{s},v^{l}_{i}) \right) \leq \sum_{i=1}^{N}\mbox{area}\left(U(v^{s},v^{l}_{i})\right)
\]
\normalsize

It follows that the probability this event occurs on some subset $I$ is bounded above by $\frac{1}{k'}\sum_{i=1}^{N}\frac{\mbox{area}\left(U(v^{s},v^{l}_{i})\right)}{\mbox{area}(B)} = \frac{1}{k'}\sum_{i=1}^{N}\left(1-\frac{2}{\pi}\cos^{-1}\left(\left\Vert v^{s}\vphantom{v^{l}} \right\Vert _{2} / \left\Vert v^{l}_{i} \right\Vert _{2}\right)\right)$. 

If $k' = N$, we use the fact that $\mbox{area}\left(\bigcap_{i\in [N]} U(v^{s},v^{l}_{i}) \right) \leq \min_{i\in[N]}\left\{ \mbox{area}\left( U(v^{s},v^{l}_{i}) \right) \right\}$ to obtain the desired result. 
\end{proof}

Intuitively, if this event occurs, then there are at least $k'$ vectors that rank above $v^{s}$ when sorted in nondecreasing order by their lengths under projection. This can only occur when the endpoint of $u$ falls in a region on the unit sphere corresponding to $\bigcup_{I \subseteq [N]:|I|=k'}\bigcap_{i\in I} U(v^{s},v^{l}_{i})$. We illustrate this region in Figure \ref{fig:3d_proj_two_vec} for the case of $d = 3$. 

\begin{thm}
\label{thm:multi_order_under_proj}
Let $\left\{ v^{l}_{i}\right\} _{i=1}^{N}$ and $\left\{ v^{s}_{i'} \vphantom{v^{l}_{i}} \right\} _{i'=1}^{N'}$ be sets of vectors such that $\left\Vert v^{l}_{i}\right\Vert _{2} > \left\Vert v^{s}_{i'}\right\Vert _{2}\forall i\in[N],i'\in[N']$.
Then the probability that there is a subset of $k'$ vectors from $\left\{ v^{l}_{i}\right\} _{i=1}^{N}$ that are all not longer than some $v^{s}_{i'}$ under projection is at most $\frac{1}{k'}\sum_{i=1}^{N}\left(1-\frac{2}{\pi}\cos^{-1}\left(\left\Vert v^{s}_{\mathrm{max}}\vphantom{v^{l}}\right\Vert _{2} / \left\Vert v^{l}_{i}\right\Vert _{2}\right)\right)$, where $\left\Vert v^{s}_{\mathrm{max}}\vphantom{v^{l}}\right\Vert _{2}\geq\left\Vert v^{s}_{i'}\right\Vert _{2}\forall i'\in[N']$. 
\end{thm}

\begin{proof}
The probability that this event occurs is at most $\mathrm{Pr}\left(u\in\bigcup_{i'\in[N']}\bigcup_{I\subseteq[N]:|I|=k'}\bigcap_{i\in I} U(v^{s}_{i'},v^{l}_{i})\right)$. We observe that for all $i,i'$, $\left\{ \theta \left| \left|\cos\theta\right| \leq \left\Vert v^{s}_{i'} \right\Vert _{2} / \left\Vert v^{l}_{i} \right\Vert _{2} \right.\right\} \subseteq \left\{ \theta \left| \left|\cos\theta\right| \leq \left\Vert v^{s}_{\mathrm{max}}\vphantom{v^{l}}\right\Vert _{2} / \left\Vert v^{l}_{i} \right\Vert _{2} \right.\right\}$, which implies that $U(v^{s}_{i'},v^{l}_{i}) \subseteq U(v^{s}_{\mathrm{max}},v^{l}_{i})$. 

If we take the intersection followed by union on both sides, we obtain $\bigcup_{I\subseteq[N]:|I|=k'}\bigcap_{i\in I} U(v^{s}_{i'},v^{l}_{i}) \subseteq \bigcup_{I\subseteq[N]:|I|=k'}\bigcap_{i\in I} U(v^{s}_{\mathrm{max}},v^{l}_{i})$. Because this is true for all $i'$, $\bigcup_{i'\in[N']}\bigcup_{I\subseteq[N]:|I|=k'}\bigcap_{i\in I} U(v^{s}_{i'},v^{l}_{i}) \subseteq\bigcup_{I\subseteq[N]:|I|=k'}\bigcap_{i\in I} U(v^{s}_{\mathrm{max}},v^{l}_{i})$. 

Therefore, this probability is bounded above by $\mathrm{Pr}\left(u\in\bigcup_{I\subseteq[N]:|I|=k'}\bigcap_{i\in I} U(v^{s}_{\mathrm{max}},v^{l}_{i}) \right)$. By Lemma \ref{lem:single_order_under_proj}, this is at most $\frac{1}{k'}\sum_{i=1}^{N}\left(1-\frac{2}{\pi}\cos^{-1}\left(\left\Vert v^{s}_{\mathrm{max}}\vphantom{v^{l}}\right\Vert _{2} / \left\Vert v^{l}_{i}\right\Vert _{2}\right)\right)$. 

\end{proof}

\subsection{Data Density}

We now formally characterize data density by defining the following notion of local relative sparsity:

\begin{defn}
Given a dataset $D\subseteq\mathbb{R}^{d}$, let $B_{p}(r)$ be the set of points in $D$ that are within a ball of radius $r$ around a point $p$. We say $D$ has local relative sparsity of $(\tau,\gamma)$ at a point $p \in \mathbb{R}^{d}$ if for all $r$ such that $\left|B_{p}(r)\right|\geq\tau$, $\left|B_{p}(\gamma r)\right|\leq2\left|B_{p}(r)\right|$, where $\gamma \geq 1$. 
\end{defn}
Intuitively, $\gamma$ represents a lower bound on the increase in radius when the number of points within the ball is doubled. When $\gamma$ is close to $1$, the dataset is dense in the neighbourhood of $p$, since there could be many points in $D$ that are almost equidistant from $p$. Retrieving the nearest neighbours of such a $p$ is considered ``hard'', since it would be difficult to tell which of these points are the true nearest neighbours without computing the distances to all these points exactly. 

We also define a related notion of global relative sparsity, which we will use to derive the number of iterations the outer loop of the querying function should be executed and a bound on the running time that is independent of the query:

\begin{defn}
A dataset $D$ has global relative sparsity of $(\tau,\gamma)$ if for all $r$ and $p \in \mathbb{R}^{d}$ such that $\left|B_{p}(r)\right|\geq\tau$, $\left|B_{p}(\gamma r)\right|\leq2\left|B_{p}(r)\right|$, where $\gamma \geq 1$. 
\end{defn}

Note that a dataset with global relative sparsity of $(\tau,\gamma)$ has local relative sparsity of $(\tau,\gamma)$ at every point. Global relative sparsity is closely related to the notion of \emph{expansion rate} introduced by \cite{karger2002finding}. More specifically, a dataset with global relative sparsity of $(\tau,\gamma)$ has $(\tau,2^{(1/\log_{2}\gamma)})$-expansion, where the latter quantity is known as the expansion rate. If we use $c$ to denote the expansion rate, the quantity $\log_{2}c$ is known as the \emph{intrinsic dimension}~\cite{clarkson2006nearest}, since when the dataset is uniformly distributed, intrinsic dimensionality would match ambient dimensionality. So, the intrinsic dimension of a dataset with global relative sparsity of $(\tau,\gamma)$ is $1/\log_{2}\gamma$. 

\subsection{Data-Independent Version}

In the data-independent version of the algorithm, the outer loop in the querying function executes for a preset number of iterations $\tilde{k}$. The values of $L$, $m$ and $\tilde{k}$ are fixed for all queries and will be chosen later. 

We apply the results obtained above to analyze the algorithm. Consider the event that the algorithm fails to return the correct set of $k$-nearest neighbours -- this can only occur if a true $k$-nearest neighbour is not contained in any of the $S_{l}$'s, which entails that for each $l \in [L]$, there is a set of $\tilde{k} - k + 1$ points that are not the true $k$-nearest neighbours but are closer to the query than the true $k$-nearest neighbours under some of the projections $u_{1l}, \ldots, u_{ml}$. We analyze the probability that this occurs below and derive the parameter settings that ensure the algorithm succeeds with high probability. Please refer to the supplementary material for proofs of the following results. 

\begin{lem}
\label{lem:single_proj_failure_prob}
For a dataset with global relative sparsity $(k,\gamma)$, there is some $\tilde{k} \in \Omega(\max(k\log (n/k),k(n/k)^{1-\log_{2}\gamma}))$ such that the probability that the candidate points retrieved from a given composite index do not include some of the true $k$-nearest neighbours is at most some constant $\alpha < 1$.
\end{lem}

\begin{thm}
\label{thm:data_indep_alg_correctness}
For a dataset with global relative sparsity $(k,\gamma)$, for any $\epsilon > 0$, there is some $L$ and $\tilde{k} \in \Omega(\max(k\log (n/k),k(n/k)^{1-\log_{2}\gamma}))$ such that the algorithm returns the correct set of $k$-nearest neighbours with probability of at least $1 - \epsilon$. 
\end{thm}

The above result suggests that we should choose $\tilde{k} \in \Omega(\max(k\log (n/k),k(n/k)^{1-\log_{2}\gamma}))$ to ensure the algorithm succeeds with high probability. Next, we analyze the time and space complexity of the algorithm. Proofs of the following results are found in the supplementary material. 

\begin{thm}
\label{thm:data_indep_alg_query_time_complexity}
The algorithm takes $O(\max(dk\log (n/k),dk(n/k)^{1-1 / d'}))$ time to retrieve the $k$-nearest neighbours at query time, where $d'$ denotes the intrinsic dimension of the dataset. 
\end{thm}

\begin{thm}
\label{thm:data_indep_alg_construction_time_complexity}
The algorithm takes $O(dn+n\log n)$ time to preprocess the data points in $D$ at construction time. 
\end{thm}

\begin{thm}
\label{thm:data_indep_alg_update_time_complexity}
The algorithm requires $O(d+\log n)$ time to insert a new data point and $O(\log n)$ time to delete a data point. 
\end{thm}

\begin{thm}
\label{thm:data_indep_alg_space_complexity}
The algorithm requires $O(n)$ space in addition to the space used to store the data. 
\end{thm}

\subsection{Data-Dependent Version}

Conceptually, performance of the proposed algorithm depends on two factors: how likely the index returns the true nearest neighbours before other points and when the algorithm stops retrieving points from the index. The preceding sections primarily focused on the former; in this section, we take a closer look at the latter. 

One strategy, which is used by the data-independent version of the algorithm, is to stop after a preset number of iterations of the outer loop. Although simple, such a strategy leaves much to be desired. First of all, in order to set the number of iterations, it requires knowledge of the global relative sparsity of the dataset, which is rarely known a priori. Computing this is either very expensive in the case of  datasets or infeasible in the case of streaming data, as global relative sparsity may change as new data points arrive. More importantly, it is unable to take advantage of the local relative sparsity in the neighbourhood of the query. A method that is capable of adapting to local relative sparsity could potentially be much faster because query points tend to be close to the manifold on which points in the dataset lie, resulting in the dataset being sparse in the neighbourhood of the query point. 

Ideally, the algorithm should stop as soon as it has retrieved the true nearest neighbours. Determining if this is the case amounts to asking if there exists a point that we have not seen lying closer to the query than the points we have seen. At first sight, because nothing is known about unseen points, it seems not possible to do better than exhaustive search, as we can only rule out the existence of such a point after computing distances to all unseen points. Somewhat surprisingly, by exploiting the fact that the projections associated with the index are random, it is possible to make inferences about points that we have never seen. We do so by leveraging ideas from statistical hypothesis testing. 

After each iteration of the outer loop, we perform a hypothesis test, with the null hypothesis being that the complete set of the $k$-nearest neighbours has not yet been retrieved. Rejecting the null hypothesis implies accepting the alternative hypothesis that all the true $k$-nearest neighbours have been retrieved. At this point, the algorithm can safely terminate while guaranteeing that the probability that the algorithm fails to return the correct results is bounded above by the significance level. The test statistic is an upper bound on the probability of missing a true $k$-nearest neighbour. The resulting algorithm does not require any prior knowledge about the dataset and terminates earlier when the dataset is sparse in the neighbourhood of the query; for this reason, we will refer to this version of the algorithm as the data-dependent version. 

More concretely, as the algorithm retrieves candidate points, it computes their true distances to the query and maintains a list of $k$ points that are the closest to the query among the points retrieved from all composite indices so far. Let $\tilde{p}^{(i)}$ and $\tilde{p}_{l}^{\mathrm{max}}$ denote the $i^{\mathrm{th}}$ closest candidate point to $q$ retrieved from all composite indices and the farthest candidate point from $q$ retrieved from the $l^{\mathrm{th}}$ composite index respectively. When the number of candidate points exceeds $k$, the algorithm checks if $\prod_{l=1}^{L}\left(1-\left(\frac{2}{\pi}\cos^{-1}\left(\left\Vert \tilde{p}^{(k)}-q\right\Vert _{2}/\left\Vert \tilde{p}_{l}^{\mathrm{max}} \vphantom{\tilde{p}^{(k)}} -q\right\Vert _{2}\right)\right)^{m}\right)\leq\epsilon$, where $\epsilon$ is the maximum tolerable failure probability, after each iteration of the outer loop. If the condition is satisfied, the algorithm terminates and returns $\left\{ \tilde{p}^{(i)}\right\} _{i=1}^{k}$. 

We show the correctness and running time of this algorithm below. Proofs of the following results are found in the supplementary material. 

\begin{thm}
\label{thm:data_dep_alg_correctness}
For any $\epsilon > 0$, $m$ and $L$, the data-dependent algorithm returns the correct set of $k$-nearest neighbours of the query $q$ with probability of at least $1 - \epsilon$. 
\end{thm}

\begin{thm}
\label{thm:data_dep_alg_time}
On a dataset with global relative sparsity $(k,\gamma)$, given fixed parameters $m$ and $L$, the data-dependent algorithm takes \small$O\left(\max\left(dk\log\left(\frac{n}{k}\right),dk\left(\frac{n}{k}\right)^{1-1/d'},\frac{d}{\left(1-\sqrt[m]{1-\sqrt[L]{\epsilon}}\right)^{d'}}\right)\right)$\normalsize \; time with high probability to retrieve the $k$-nearest neighbours at query time, where $d'$ denotes the intrinsic dimension of the dataset. 
\end{thm}

Note that we can make the denominator of the last argument arbitrarily close to $1$ by choosing a large $L$. 

\section{Experiments}

\begin{figure*}[t]
    \centering
    \subfloat[]{
        \includegraphics[width=0.33\textwidth]{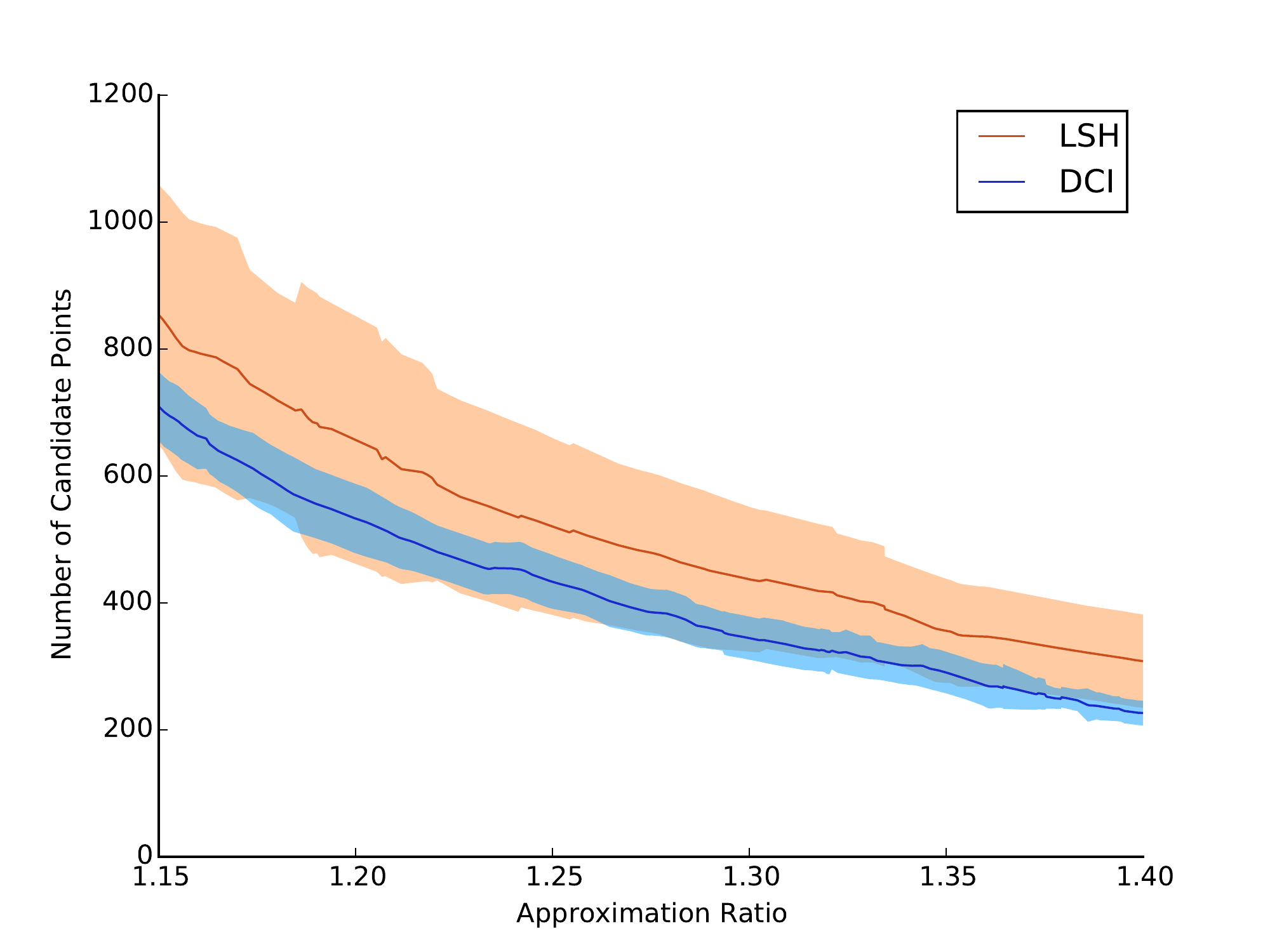}
        \label{fig:plot_cifar_perf}
    }
    \subfloat[]{
        \includegraphics[width=0.33\textwidth]{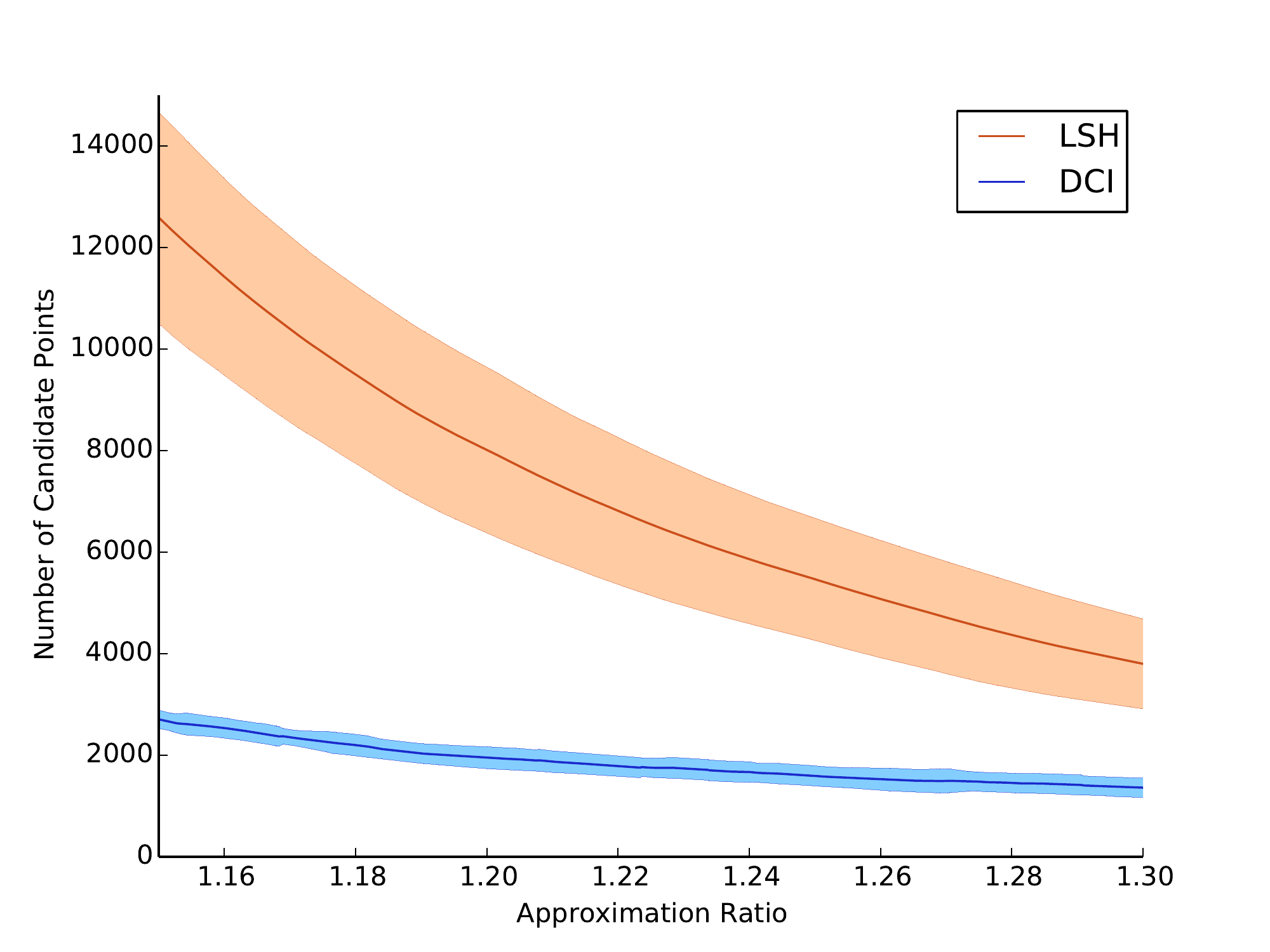}
        \label{fig:plot_mnist_perf}
    }
    \subfloat[]{
        \includegraphics[width=0.33\textwidth]{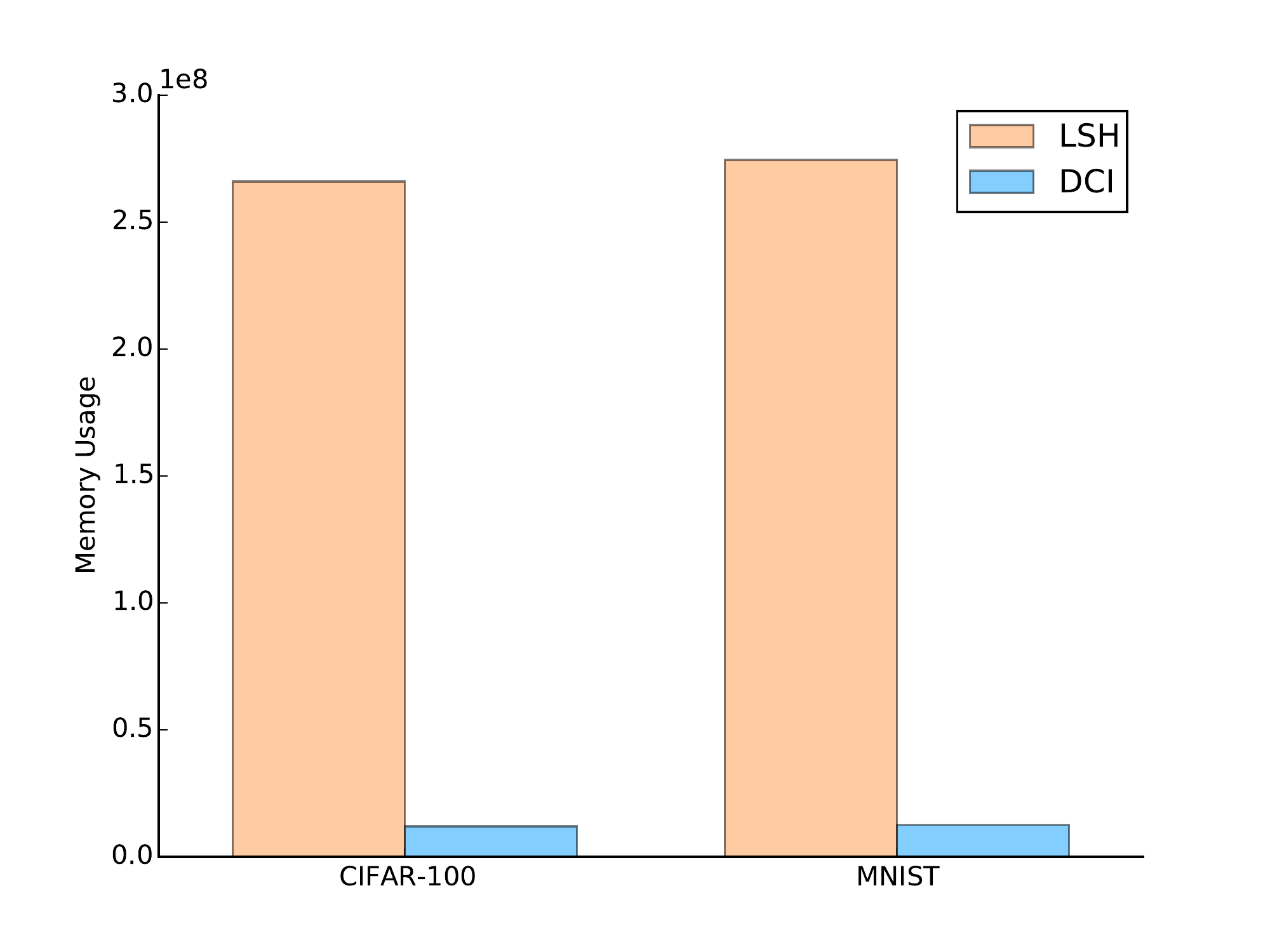}
        \label{fig:plot_mem_comp}
    }
    \caption{(a-b) Comparison of the speed of the proposed method (labelled as DCI) and LSH on (a) CIFAR-100 and (b) MNIST. Each curve represents the mean number to candidate points required to achieve varying levels of approximation quality over ten folds and the shaded area represents $\pm$1 standard deviation. (c) Comparison of the space efficiency of the proposed method and LSH on CIFAR-100 and MNIST. The height of each bar represents the average amount of memory used by each method to achieve the performance shown in (a) and (b). }
\end{figure*}

We compare the performance of the proposed algorithm to that of LSH, which is arguably the most popular method for fast nearest neighbour retrieval. Because LSH is designed for the approximate setting, under which the performance metric of interest is how close the points returned by the algorithm are to the query rather than whether the returned points are the true $k$-nearest neighbours, we empirically evaluate performance in this setting. Because the distance metric of interest is the Euclidean distance, we compare to Exact Euclidean LSH ($\mathrm{E^{2}LSH}$)~\cite{datar2004locality}, which uses hash functions designed for Euclidean space. 

We compare the performance of the proposed algorithm and LSH on the CIFAR-100 and MNIST datasets, which consist of $32 \times 32$ colour images of various real-world objects and $28 \times 28$ grayscale images of handwritten digits respectively. We reshape the images into vectors, with each dimension representing the pixel intensity at a particular location and colour channel of the image. The resulting vectors have a dimensionality of $32 \times 32 \times 3 = 3072$ in the case of CIFAR-100 and $28 \times 28 = 784$ in the case of MNIST, so the dimensionality under consideration is higher than what traditional tree-based algorithms can handle. We combine the training set and the test set of each dataset, and so we have a total of $60,000$ instances in CIFAR-100 and $70,000$ instances in MNIST. 

We randomize the instances that serve as queries using cross-validation. Specifically, we randomly select 100 instances to serve as query points and the designate the remaining instances as data points. Each algorithm is then used to retrieve approximate $k$-nearest neighbours of each query point among the set of all data points. This procedure is repeated for ten folds, each with a different split of query vs. data points. 

We compare the number of candidate points that each algorithm requires to achieve a desired level of approximation quality. We quantify approximation quality using the approximation ratio, which is defined as the ratio of the radius of the ball containing approximate $k$-nearest neighbours to the radius of the ball containing true $k$-nearest neighbours. So, the smaller the approximation ratio, the better the approximation quality. Because dimensionality is high and exhaustive search must be performed over all candidate points, the time taken for compute distances between the candidate points and the query dominates the overall running time of the querying operation. Therefore, the number of candidate points can be viewed as an implementation-independent proxy for the running time. 

Because the hash table constructed by LSH depends on the desired level of approximation quality, we construct a different hash table for each level of approximation quality. On the other hand, the indices constructed by the proposed method are not specific to any particular level of approximation quality; instead, approximation quality can be controlled at query time by varying the number of iterations of the outer loop. So, the same indices are used for all levels of approximation quality. Therefore, our evaluation scheme is biased towards LSH and against the proposed method. 

We adopt the recommended guidelines for choosing parameters of LSH and used $24$ hashes per table and $100$ tables. For the proposed method, we used $m = 25$ and $L = 2$ on CIFAR-100 and $m = 15$ and $L = 3$ on MNIST, which we found to work well in practice. In Figures \ref{fig:plot_cifar_perf} and \ref{fig:plot_mnist_perf}, we plot the performance of the proposed method and LSH on the CIFAR-100 and MNIST datasets for retrieving $25$-nearest neighbours. 

For the purposes of retrieving nearest neighbours, MNIST is a more challenging dataset than CIFAR-100. This is because the instances in MNIST form dense clusters, whereas the instances in CIFAR-100 are more visually diverse and so are more dispersed in space. Intuitively, if the query falls inside a dense cluster of points, there are many points that are very close to the query and so it is difficult to distinguish true nearest neighbours from points that are only slightly farther away. Viewed differently, because true nearest neighbours tend to be extremely close to the query on MNIST, the denominator for computing the approximation ratio is usually very small. Consequently, returning points that are only slightly farther away than the true nearest neighbours would result in a large approximation ratio. As a result, both the proposed method and LSH require far more candidate points on MNIST than on CIFAR-100 to achieve comparable approximation ratios. 

We find the proposed method achieves better performance than LSH at all levels of approximation quality. Notably, the performance of LSH degrades drastically on MNIST, which is not surprising since LSH is known to have difficulties on datasets with large variations in data density. On the other hand, the proposed method requires $61.3\% - 78.7\%$ fewer candidate points than LSH to achieve the same approximation quality while using less than $1/20$ of the memory. 

\section{Conclusion}

In this paper, we delineated the inherent deficiencies of space partitioning and presented a new strategy for fast retrieval of $k$-nearest neighbours, which we dub dynamic continuous indexing (DCI). Instead of discretizing the vector space, the proposed algorithm constructs continuous indices, each of which imposes an ordering of data points in which closeby positions approximately reflect proximity in the vector space. Unlike existing methods, the proposed algorithm allows granular control over accuracy and speed on a per-query basis, adapts to variations to data density on-the-fly and supports online updates to the dataset. We analyzed the proposed algorithm and showed it runs in time linear in ambient dimensionality and sub-linear in intrinsic dimensionality and the size of the dataset and takes space constant in ambient dimensionality and linear in the size of the dataset. Furthermore, we demonstrated empirically that the proposed algorithm compares favourably to LSH in terms of approximation quality, speed and space efficiency. 

\paragraph{Acknowledgements.} Ke Li thanks the Berkeley Vision and Learning Center (BVLC) and the Natural Sciences and Engineering Research Council of Canada (NSERC) for financial support. The authors also thank Lap Chi Lau and the anonymous reviewers for feedback. 

\bibliography{fast_knn}
\bibliographystyle{icml2016}

\newpage

\twocolumn[
\icmltitle{Fast $k$-Nearest Neighbour Search via Dynamic Continuous Indexing \\ \vspace{5pt} \large Supplementary Material}

\icmlauthor{Ke Li}{ke.li@eecs.berkeley.edu}
\icmlauthor{Jitendra Malik}{malik@eecs.berkeley.edu}
\icmladdress{University of California, Berkeley, CA 94720, United States}

\vskip 0.3in
]

Below, we present proofs of the results shown in the paper. We first prove two intermediate results, which are used to derive results in the paper. Throughout our proofs, we use $\{p^{(i)}\}_{i=1}^{n}$ to denote a re-ordering of the points $\{p^{i}\}_{i=1}^{n}$ so that $p^{(i)}$ is the $i^{\mathrm{th}}$ closest point to the query $q$. For any given projection direction $u_{jl}$ associated with a simple index, we also consider a ranking of the points $\{p^{i}\}_{i=1}^{n}$ by their distance to $q$ under projection $u_{jl}$ in nondecreasing order. We say points are ranked before others if they appear earlier in this ranking. 

\begin{lem}
\label{lem:prob_num_extraneous_points}
The probability that for all constituent simple indices of a composite index, fewer than $n_{0}$ points exist that are not the true $k$-nearest neighbours but are ranked before some of them, is at least $\left[1 - \frac{1}{n_{0} - k}\sum_{i=2k+1}^{n}\frac{\left\Vert p^{(k)}-q\right\Vert _{2}}{\left\Vert p^{(i)}-q\right\Vert _{2}}\right]^{m}$. 
\end{lem}

\begin{proof}
For any given simple index, we will refer to the points that are not the true $k$-nearest neighbours but are ranked before some of them as \emph{extraneous points}. We furthermore categorize the extraneous points as either \emph{reasonable} or \emph{silly}. An extraneous point is reasonable if it is one of the $2k$-nearest neighbours, and is silly otherwise. Since there can be at most $k$ reasonable extraneous points, there must be at least $n_{0} - k$ silly extraneous points. Therefore, the event that $n_{0}$ extraneous points exist must be contained in the event that $n_{0} - k$ silly extraneous points exist. 

We find the probability that such a set of silly extraneous points exists for any given simple index. By Theorem \ref{thm:multi_order_under_proj}, where we take $\{v^{s}_{i'}\}_{i'=1}^{N'}$ to be $\{p^{(i)} - q\}_{i=1}^{k}$, $\{v^{l}_{i}\}_{i=1}^{N}$ to be $\{p^{(i)} - q\}_{i=2k+1}^{n}$ and $k'$ to be $n_{0} - k$, the probability that there are at least $n_{0} - k$ silly extraneous points is at most $\frac{1}{n_{0} - k}\sum_{i=2k+1}^{n}\left(1-\frac{2}{\pi}\cos^{-1}\left(\frac{\left\Vert p^{(k)} - q\right\Vert _{2}}{\left\Vert p^{(i)} - q\right\Vert _{2}}\right)\right)$. This implies that the probability that at least $n_{0}$ extraneous points exist is bounded above by the same quantity, and so the probability that fewer than $n_{0}$ extraneous points exist is at least $1 - \frac{1}{n_{0} - k}\sum_{i=2k+1}^{n}\left(1-\frac{2}{\pi}\cos^{-1}\left(\frac{\left\Vert p^{(k)} - q\right\Vert _{2}}{\left\Vert p^{(i)} - q\right\Vert _{2}}\right)\right)$. Hence, the probability that fewer than $n_{0}$ extraneous points exist for all constituent simples indices of a composite index is at least $\left[1 - \frac{1}{n_{0} - k}\sum_{i=2k+1}^{n}\left(1-\frac{2}{\pi}\cos^{-1}\left(\frac{\left\Vert p^{(k)} - q\right\Vert _{2}}{\left\Vert p^{(i)} - q\right\Vert _{2}}\right)\right)\right]^{m}$. Using the fact that $1-(2 / \pi)\cos^{-1}\left(x\right)\leq x \; \forall x \in [0,1]$, this quantity is at least $\left[1 - \frac{1}{n_{0} - k}\sum_{i=2k+1}^{n}\frac{\left\Vert p^{(k)}-q\right\Vert _{2}}{\left\Vert p^{(i)}-q\right\Vert _{2}}\right]^{m}$. 
\end{proof}

\begin{lem}
\label{lem:prob_num_extraneous_points_sparsity}
On a dataset with global relative sparsity $(k,\gamma)$, the probability that for all constituent simple indices of a composite index, fewer than $n_{0}$ points exist that are not the true $k$-nearest neighbours but are ranked before some of them, is at least $\left[1 - \frac{1}{n_{0} - k}O\left(\max(k\log (n/k),k(n/k)^{1-\log_{2}\gamma})\right)\right]^{m}$. 
\end{lem}
\begin{proof}
By definition of global relative sparsity, for all $i\geq2k+1$, $\left\Vert p^{(i)}-q\right\Vert _{2}>\gamma\left\Vert p^{(k)}-q\right\Vert _{2}$. By applying this recursively, we see that for all $i\geq2^{i'}k+1$, $\left\Vert p^{(i)}-q\right\Vert _{2}>\gamma^{i'}\left\Vert p^{(k)}-q\right\Vert _{2}$. It follows that $\sum_{i=2k+1}^{n}\frac{\left\Vert p^{(k)}-q\right\Vert _{2}}{\left\Vert p^{(i)}-q\right\Vert _{2}}$ is less than $\sum_{i'=1}^{\lceil\log_{2}(n/k)\rceil-1}2^{i'}k\gamma^{-i'}$. If $\gamma \geq 2$, this quantity is at most $k\log_{2}\left(\frac{n}{k}\right)$. If $1 \leq\gamma < 2$, this quantity is:

\vspace{-10pt}
\footnotesize
\begin{align*}
 & k\left(\frac{2}{\gamma}\right)\left(\left(\frac{2}{\gamma}\right)^{\lceil\log_{2}(n/k)\rceil-1}-1\right)/\left(\frac{2}{\gamma}-1\right) \\
 = \;& O\left(k\left(\frac{2}{\gamma}\right)^{\lceil\log_{2}(n/k)\rceil-1}\right) \\
 = \;& O\left(k\left(\frac{n}{k}\right)^{1-\log_{2}\gamma}\right)
\end{align*}
\normalsize

Combining this bound with Lemma \ref{lem:prob_num_extraneous_points} yields the desired result. 

\end{proof}

\begin{customlem}{6}
For a dataset with global relative sparsity $(k,\gamma)$, there is some $\tilde{k} \in \Omega(\max(k\log (n/k),k(n/k)^{1-\log_{2}\gamma}))$ such that the probability that the candidate points retrieved from a given composite index do not include some of the true $k$-nearest neighbours is at most some constant $\alpha < 1$.
\end{customlem}

\begin{proof}
We will refer to points ranked in the top $\tilde{k}$ positions that are the true $k$-nearest neighbours as \emph{true positives} and those that are not as \emph{false positives}. Additionally, we will refer to points not ranked in the top $\tilde{k}$ positions that are the true $k$-nearest neighbours as \emph{false negatives}. 

When not all the true $k$-nearest neighbours are in the top $\tilde{k}$ positions, then there must be at least one false negative. Since there are at most $k-1$ true positives, there must be at least $\tilde{k}-(k-1)$ false positives. 

Since false positives are not the true $k$-nearest neighbours but are ranked before the false negative, which is a true $k$-nearest neighbour, we can apply Lemma \ref{lem:prob_num_extraneous_points_sparsity}. By taking $n_{0}$ to be $\tilde{k}-(k-1)$, we obtain a lower bound on the probability of the existence of fewer than $\tilde{k}-(k-1)$ false positives for all constituent simple indices of the composite index, which is $\left[1 - \frac{1}{\tilde{k}-2k+1}O\left(\max(k\log (n/k),k(n/k)^{1-\log_{2}\gamma})\right)\right]^{m}$. If each simple index has fewer than $\tilde{k}-(k-1)$ false positives, then the top $\tilde{k}$ positions must contain all the true $k$-nearest neighbours. Since this is true for all constituent simple indices, all the true $k$-nearest neighbours must be among the candidate points after $\tilde{k}$ iterations of the outer loop. The failure probability is therefore at most $1 - \left[1 - \frac{1}{\tilde{k}-2k+1}O\left(\max(k\log (n/k),k(n/k)^{1-\log_{2}\gamma})\right)\right]^{m}$. So, there is some $\tilde{k} \in \Omega(\max(k\log (n/k),k(n/k)^{1-\log_{2}\gamma}))$ that makes this quantity strictly less than 1. 

\end{proof}

\begin{customthm}{7}
For a dataset with global relative sparsity $(k,\gamma)$, for any $\epsilon > 0$, there is some $L$ and $\tilde{k} \in \Omega(\max(k\log (n/k),k(n/k)^{1-\log_{2}\gamma}))$ such that the algorithm returns the correct set of $k$-nearest neighbours with probability of at least $1 - \epsilon$. 
\end{customthm}

\begin{proof}
By Lemma \ref{lem:single_proj_failure_prob}, the first $\tilde{k}$ points retrieved from a given composite index do not include some of the true $k$-nearest neighbours with probability of at most $\alpha$. For the algorithm to fail, this must occur for all composite indices. Since each composite index is constructed independently, the algorithm fails with probability of at most $\alpha^{L}$, and so must succeed with probability of at least $1-\alpha^{L}$. Since $\alpha < 1$, there is some $L$ that makes $1-\alpha^{L} \geq 1 - \epsilon$. 
\end{proof}

\begin{customthm}{8}
The algorithm takes $O(\max(dk\log (n/k),dk(n/k)^{1-1 / d'}))$ time to retrieve the $k$-nearest neighbours at query time, where $d'$ denotes the intrinsic dimension of the dataset. 
\end{customthm}

\begin{proof}
Computing projections of the query point along all $u_{jl}$'s takes $O(d)$ time, since $m$ and $L$ are constants. Searching in the binary search trees/skip lists $T_{jl}$'s takes $O(\log n)$ time. The total number of candidate points retrieved is at most $\Theta(\max(k\log (n/k),k(n/k)^{1-\log_{2}\gamma}))$. Computing the distance between each candidate point and the query point takes at most $O(\max(dk\log (n/k),dk(n/k)^{1-\log_{2}\gamma}))$ time. We can find the $k$ closest points to $q$ in the set of candidate points using a selection algorithm like quickselect, which takes $O(\max(k\log (n/k),k(n/k)^{1-\log_{2}\gamma}))$ time on average. Since the time taken to compute distances to the query point dominates, the entire algorithm takes $O(\max(dk\log (n/k),dk(n/k)^{1-\log_{2}\gamma}))$ time. Since $d' = 1/\log_{2}\gamma$, this can be rewritten as $O(\max(dk\log (n/k),dk(n/k)^{1-1 / d'}))$. 
\end{proof}

\begin{customthm}{9}
The algorithm takes $O(dn+n\log n)$ time to preprocess the data points in $D$ at construction time. 
\end{customthm}

\begin{proof}
Computing projections of all $n$ points along all $u_{jl}$'s takes $O(dn)$ time, since $m$ and $L$ are constants. Inserting all $n$ points into $mL$ self-balancing binary search trees/skip lists takes $O(n\log n)$ time. 
\end{proof}

\begin{customthm}{10}
The algorithm requires $O(d+\log n)$ time to insert a new data point and $O(\log n)$ time to delete a data point. 
\end{customthm}

\begin{proof}
In order to insert a data point, we need to compute its projection along all $u_{jl}$'s and insert it into each binary search tree or skip list. Computing the projection takes $O(d)$ time and inserting an entry into a self-balancing binary search tree or skip list takes $O(\log n)$ time. In order to delete a data point, we simply remove it from each of the binary search trees or skip lists, which takes $O(\log n)$ time. 
\end{proof}

\begin{customthm}{11}
The algorithm requires $O(n)$ space in addition to the space used to store the data. 
\end{customthm}

\begin{proof}
The only additional information that needs to be stored are the $mL$ binary search trees or skip lists. Since $n$ entries are stored in each binary search tree/skip list, the additional space required is $O(n)$. 
\end{proof}

\begin{customthm}{12}
For any $\epsilon > 0$, $m$ and $L$, the data-dependent algorithm returns the correct set of $k$-nearest neighbours of the query $q$ with probability of at least $1 - \epsilon$. 
\end{customthm}

\begin{proof}
We analyze the probability that the algorithm fails to return the correct set of $k$-nearest neighbours. Let $p^{*}$ denote a true $k$-nearest neighbour that was missed. If the algorithm fails, then for any given composite index, $p^{*}$ is not among the candidate points retrieved from the said index. In other words, the composite index must have returned all these points before $p^{*}$, implying that at least one constituent simple index returns all these points before $p^{*}$. This means that all these points must appear closer to $q$ than $p^{*}$ under the projection associated with the simple index. By Lemma \ref{lem:single_order_under_proj}, if we take $\left\{ v^{l}_{i}\right\} _{i=1}^{N}$ to be displacement vectors from $q$ to the candidate points that are farther from $q$ than $p^{*}$ and $v^{s}$ to be the displacement vector from $q$ to $p^{*}$, the probability of this occurring for a given constituent simple index of the $l^{\mathrm{th}}$ composite index is at most $1-\frac{2}{\pi}\cos^{-1}\left(\left\Vert p^{*}-q\right\Vert _{2}/\left\Vert \tilde{p}_{l}^{\mathrm{max}}-q\right\Vert _{2}\right)$. The probability that this occurs for \emph{some} constituent simple index is at most $1-\left(\frac{2}{\pi}\cos^{-1}\left(\left\Vert p^{*}-q\right\Vert _{2}/\left\Vert \tilde{p}_{l}^{\mathrm{max}}-q\right\Vert _{2}\right)\right)^{m}$. For the algorithm to fail, this must occur for all composite indices; so the failure probability is at most $\prod_{l=1}^{L}\left(1-\left(\frac{2}{\pi}\cos^{-1}\left(\left\Vert p^{*}-q\right\Vert _{2}/\left\Vert \tilde{p}_{l}^{\mathrm{max}}-q\right\Vert _{2}\right)\right)^{m}\right)$. 

We observe that $\left\Vert p^{*} \vphantom{\tilde{p}^{(k)}} -q\right\Vert _{2} \leq \left\Vert p^{(k)}-q\right\Vert _{2} \leq \left\Vert \tilde{p}^{(k)}-q\right\Vert _{2}$ since there can be at most $k-1$ points in the dataset that are closer to $q$ than $p^{*}$. So, the failure probability can be bounded above by $\prod_{l=1}^{L}\left(1-\left(\frac{2}{\pi}\cos^{-1}\left(\left\Vert \tilde{p}^{(k)}-q\right\Vert _{2}/\left\Vert \tilde{p}_{l}^{\mathrm{max}} \vphantom{\tilde{p}^{(k)}} -q\right\Vert _{2}\right)\right)^{m}\right)$. When the algorithm terminates, we know this quantity is at most $\epsilon$. Therefore, the algorithm returns the correct set of $k$-nearest neighbours with probability of at least $1 - \epsilon$. 
\end{proof}

\begin{customthm}{13}
On a dataset with global relative sparsity $(k,\gamma)$, given fixed parameters $m$ and $L$, the data-dependent algorithm takes \small$O\left(\max\left(dk\log\left(\frac{n}{k}\right),dk\left(\frac{n}{k}\right)^{1-1/d'},\frac{d}{\left(1-\sqrt[m]{1-\sqrt[L]{\epsilon}}\right)^{d'}}\right)\right)$\normalsize \; time with high probability to retrieve the $k$-nearest neighbours at query time, where $d'$ denotes the intrinsic dimension of the dataset. 
\end{customthm}

\begin{proof}
In order to bound the running time, we bound the total number of candidate points retrieved until the stopping condition is satisfied. We divide the execution of the algorithm into two stages and analyze the algorithm's behaviour before and after it finishes retrieving all the true $k$-nearest neighbours. We first bound the number of candidate points the algorithm retrieves before finding the complete set of $k$-nearest neighbours. By Lemma \ref{lem:prob_num_extraneous_points_sparsity}, the probability that there exist fewer than $n_{0}$ points that are not the $k$-nearest neighbours but are ranked before some of them in all constituent simple indices of any given composite index is at least $\left[1 - \frac{1}{n_{0} - k}O\left(\max(k\log (n/k),k(n/k)^{1-\log_{2}\gamma})\right)\right]^{m}$. We can choose some $n_{0} \in \Theta\left(\max(k\log (n/k),k(n/k)^{1-\log_{2}\gamma})\right)$ that makes this probability arbitrarily close to 1. So, there are $\Theta\left(\max(k\log (n/k),k(n/k)^{1-\log_{2}\gamma})\right)$ such points in each constituent simple index with high probability, implying that the algorithm retrieves at most $\Theta\left(\max(k\log (n/k),k(n/k)^{1-\log_{2}\gamma})\right)$ extraneous points from any given composite index before finishing fetching all the true $k$-nearest neighbours. Since the number of composite indices is constant, the total number of candidate points retrieved from all composite indices during this stage is $k + \Theta\left(\max(k\log (n/k),k(n/k)^{1-\log_{2}\gamma})\right) = \Theta\left(\max(k\log (n/k),k(n/k)^{1-\log_{2}\gamma})\right)$ with high probability. 

After retrieving all the $k$-nearest neighbours, if the stopping condition has not yet been satisfied, the algorithm would continue retrieving points. We analyze the number of additional points the algorithm retrieves before it terminates. To this end, we bound the ratio $\left\Vert \tilde{p}^{(k)}-q\right\Vert _{2}/\left\Vert \tilde{p}_{l}^{\mathrm{max}} \vphantom{\tilde{p}^{(k)}} - q\right\Vert _{2}$ in terms of the number of candidate points retrieved so far. Since all the true $k$-nearest neighbours have been retrieved, $\left\Vert \tilde{p}^{(k)}-q\right\Vert _{2} = \left\Vert p^{(k)}-q\right\Vert _{2}$. Suppose the algorithm has already retrieved $n' - 1$ candidate points and is about to retrieve a new candidate point. Since this new candidate point must be different from any of the existing candidate points, $\left\Vert \tilde{p}_{l}^{\mathrm{max}} \vphantom{p^{(n')}} -q\right\Vert _{2} \geq \left\Vert p^{(n')}-q\right\Vert _{2}$. Hence, $\left\Vert \tilde{p}^{(k)}-q\right\Vert _{2}/\left\Vert \tilde{p}_{l}^{\mathrm{max}} \vphantom{\tilde{p}^{(k)}} - q\right\Vert _{2} \leq  \left\Vert p^{(k)} \vphantom{p^{(n')}} -q\right\Vert _{2} / \left\Vert p^{(n')}-q\right\Vert _{2}$. 

By definition of global relative sparsity, for all $n' \geq 2^{i'}k+1$, $\left\Vert p^{(n')}-q\right\Vert _{2}>\gamma^{i'}\left\Vert p^{(k)} \vphantom{p^{(n')}} -q\right\Vert _{2}$. It follows that $\left\Vert p^{(k)} \vphantom{p^{(n')}} -q\right\Vert _{2} / \left\Vert p^{(n')}-q\right\Vert _{2} < \gamma^{-\lfloor\log_{2}((n'-1)/k)\rfloor}$ for all $n'$. By combining the above inequalities, we find an upper bound on the test statistic:
\footnotesize
\begin{align*}
 &\prod_{l=1}^{L}\left(1-\left(\frac{2}{\pi}\cos^{-1}\left( \frac{\left\Vert \tilde{p}^{(k)}-q\right\Vert _{2}}{\left\Vert \tilde{p}_{l}^{\mathrm{max}} \vphantom{\tilde{p}^{(k)}} -q\right\Vert _{2}} \right)\right)^{m}\right) \\
 \leq \;& \prod_{l=1}^{L}\left(1-\left(1- \frac{\left\Vert \tilde{p}^{(k)}-q\right\Vert _{2}}{\left\Vert \tilde{p}_{l}^{\mathrm{max}} \vphantom{\tilde{p}^{(k)}} -q\right\Vert _{2}} \right)^{m}\right) \\
 < \;& \left[1-\left(1-\gamma^{-\lfloor\log_{2}((n'-1)/k)\rfloor}\right)^{m}\right]^{L} \\
 < \;& \left[1-\left(1-\gamma^{-\log_{2}((n'-1)/k)+1}\right)^{m}\right]^{L}
\end{align*}
\normalsize

Hence, if $\left[1-\left(1-\gamma^{-\log_{2}((n'-1)/k)+1}\right)^{m}\right]^{L} \leq \epsilon$, then $\prod_{l=1}^{L}\left(1-\left(\frac{2}{\pi}\cos^{-1}\left( \left\Vert \tilde{p}^{(k)}-q\right\Vert _{2}/\left\Vert \tilde{p}_{l}^{\mathrm{max}} \vphantom{\tilde{p}^{(k)}} - q\right\Vert _{2} \right)\right)^{m}\right) < \epsilon$. So, for some $n'$ that makes the former inequality true, the stopping condition would be satisfied and so the algorithm must have terminated by this point, if not earlier. By rearranging the former inequality, we find that in order for it to hold, $n'$ must be at least $2/\left(1-\sqrt[m]{1-\sqrt[L]{\epsilon}}\right)^{1/\log_{2}\gamma}$. Therefore, the number of points the algorithm retrieves before terminating cannot exceed $2/\left(1-\sqrt[m]{1-\sqrt[L]{\epsilon}}\right)^{1/\log_{2}\gamma}$. 

Combining the analysis for both stages, the number of points retrieved is at most
\footnotesize
\begin{align*}
O\left(\max\left(k\log\left(\frac{n}{k}\right),k\left(\frac{n}{k}\right)^{1-\log_{2}\gamma},\frac{1}{\left(1-\sqrt[m]{1-\sqrt[L]{\epsilon}}\right)^{\frac{1}{\log_{2}\gamma}}}\right)\right)
\end{align*}
\normalsize
with high probability. 

Since the time taken to compute distances between the query point and candidate points dominates, the running time is
\footnotesize
\begin{align*}
O\left(\max\left(dk\log\left(\frac{n}{k}\right),dk\left(\frac{n}{k}\right)^{1-\log_{2}\gamma},\frac{d}{\left(1-\sqrt[m]{1-\sqrt[L]{\epsilon}}\right)^{\frac{1}{\log_{2}\gamma}}}\right)\right)
\end{align*}
\normalsize
with high probability. 

Applying the definition of intrinsic dimension yields the desired result. 

\end{proof}

\end{document}